\documentclass{amsart}

\usepackage[a4paper,top=3cm,bottom=3cm,left=2.5cm,right=2.5cm,headsep=1em,includehead,includefoot]{geometry}

\usepackage{amssymb,amsfonts,amsmath,amsaddr}
\usepackage{amsthm,xypic,enumerate}
\usepackage{mathtools, color}
\usepackage[sort&compress,comma,square,numbers]{natbib}

\usepackage[utf8]{inputenc}
\usepackage{multirow} 
\usepackage{hyperref}

\usepackage[table]{xcolor}
\usepackage[version=3]{mhchem}
\usepackage{graphicx}
\usepackage{amssymb }
\usepackage{caption}
\usepackage{enumerate}
\usepackage{breqn}

\newtheorem{theorem}{Theorem}[section]
\newtheorem{proposition}[theorem]{Proposition}
\newtheorem{lemma}[theorem]{Lemma}

\theoremstyle{definition}

\newtheorem{remark}[theorem]{Remark}

\DeclareMathOperator{\R}{\mathbb{R}}

\DeclareMathOperator{\se}{\mathbf{c}}

\DeclareMathOperator{\diag}{diag}

\DeclareMathOperator{\HK}{HK}

\DeclareMathOperator{\im}{Im}
\DeclareMathOperator{\sign}{sign}

\newcommand\blue[1]{\textcolor{black}{#1}}

\begin{document}

\date{\today }
\title[Sign-sensitivities]{Sign-sensitivities for reaction networks: \\ an algebraic approach}
\author{Elisenda Feliu}
\address{Department of Mathematical Sciences, University of Copenhagen, \\ Universitetsparken 5, 2100 Copenhagen, Denmark}
\email{efeliu@math.ku.dk}

\maketitle

\begin{abstract}
This paper presents an algebraic framework to study \emph{sign-sensitivities} for reaction networks modeled by means of systems of ordinary differential equations. Specifically, we study the sign of the derivative of the concentrations of the species in the network at steady state with respect to a small perturbation on the parameter vector. We provide a closed formula for the derivatives that accommodates common perturbations, and illustrate its form with numerous examples. We argue that, mathematically, the study of the response to the system with respect to changes in total amounts is not well posed, and that one should rather consider perturbations with respect to the initial conditions. We find a sign-based criterion to determine, without computing the sensitivities, whether the sign depends on the steady state and parameters of the system. This is based on earlier results of so-called \emph{injective} networks. Finally, we address systems with multiple steady states and the restriction to stable steady states.

\end{abstract}

\section{Introduction}
One of the main challenges in molecular and systems biology is to infer mechanistic details of the processes that underlie available experimental data. A common strategy consists in applying controlled perturbations to the system of interest, and compare model predictions with gathered quantitative data. As an example, consider a simple two-component system in which a histidine kinase $E$ transfers a phosphate group to a response-regulator $S$:
\begin{align*}
E & \ce{->} E_p & S + E_p & \ce{ ->} S_p + E &  S_p&  \ce{->} S.
\end{align*}
Let us imagine a protein $I$, which forms an inhibitory complex $Y$ with the substrate by binding exactly one of the two forms, $S$ or $S_p$.  This gives rise to two rivaling inhibition models with the following additional reactions
\begin{align*}
\textrm{Model 1: }\quad S + I & \ce{ <=> }Y & \textrm{Model 2: }\quad S_p + I & \ce{ <=> }Y.
\end{align*}
In this small (and artificial) example, the measurement of the concentration of $I$ at steady state for two starting conditions that only differ slightly  in the amount of kinase $E$, helps us to qualitatively discriminate between these two models. Indeed, the derivative of the concentration of $I$ at steady state increases with $E$ in the first model, while it decreases in the second model. 

This exemplifies the basis of perturbation-based studies, in which the response of a system to an intervention, that being the addition of a protein, knockdown of a component, modification of reaction rate constants etc,  is recorded \cite{gardner,villaverde,paradoxical,Kholodenko:untangling}. 
In this paper we focus on how to predict this response given the model, such that comparison with experimental data can be in place. 

The modeling setting is based on reaction networks  and their associated evolution equations for the concentrations of the species in the network. In the examples, we employ  the mass-action assumption, although this is not required for the theoretical framework. 
We consider perturbations on the parameters of the model, these typically being either initial concentrations or total amounts, and kinetic parameters. We assume that the system is at steady state, and that a small perturbation to the parameter vector is performed. If the perturbation is small enough and the steady state is non-degenerate and \blue{stable}, then the system is expected to converge to a new steady state. Our goal is to determine, component-wise, the sign of the difference of the two steady states  as a function of the starting steady state and the parameters of the system. 
Mathematically, this translates into determining the signs of the derivatives of all concentrations at  steady state with respect to the performed perturbation; these signs are referred to as \emph{sign-sensitivities} \cite{SontagSigns}.

Numerous computer-based approaches exist to find sign-sensitivities, with the same or different modeling framework as the one we use here \cite{Kholodenko:untangling,gupta:sensitivity,vera:inferring}. A  general strategy assumes all parameters known  except the one of interest, and the response of the system is investigated using simulations. 
However, most quantities and parameters are notoriously hard to measure or estimate, which reduces the applicability of these methods. 
It is worth highlighting the algorithm by Sontag \cite{SontagSigns} to determine  sign-sensitivities with respect to the increase of a total amount (a quantity that is preserved under the dynamics of the system), while treating the rest of the parameters  as unknowns. 
Alternatively, if the network is small enough, one can attempt direct manipulation of the steady state equations \cite{Feliu:2010p94,bluthgen:feedback}.
In recent works \cite{okada:sensitivity,okada:law}, Okada and Mochizuki provide a theorem to determine zero sensitivities from network structure alone. Similarly, Brehm and Fiedler   study whether the sensitivity is zero \cite{brehm:sensitivity}. Another series of works \cite{shinar:sensitivity,shinar:sensitivity2} address the question of ``how large the absolute value of the sensitivity is'', by finding upper bounds for reaction networks in a specific class.

\medskip
In this paper we derive a closed formula for the sensitivities. When the kinetics are polynomial, then the derivative is expressed as a rational function in the parameters of the system and the concentrations of the species at the steady state. If the numerator and the denominator of this function have all coefficients with the same sign, then the sign of the derivative is easily determined, and it does not depend on the initial steady state nor on  the parameters of the system. When the signs of the coefficients differ, then one can employ standard techniques, such as those based on the Newton polytope, e.g. \cite{FeliuPlos}, to determine whether the derivative can both be positive and negative, thereby concluding that the sign depends on the chosen steady state and/or parameter values. 
For example, for the two rivaling inhibitory models above, we find that the derivative  $\se'_I$ of the concentration of $I$  with respect to the initial concentration of $E$ at a steady state $\se$ for 
model 1 and model 2 are
\begin{align*}
   \se_I' &=k_1k_2k_4 \se_{S}\se_{I}\, / \, q_1(k,\se) \quad \textrm{(model 1)} &    \se_I' &= -k_1k_2k_4 \se_{S}\se_{I}  \, / \, q_2(k,\se) \quad \textrm{(model 2)},
\end{align*}
where $k_i>0$ stands for the reaction rate constants of the reactions in the network (in the order given above), $\se_Z$ denotes the concentration of the species $Z$ at the steady state, and 
\begin{align*}
q_1(k,\se)&= k_{1}k_{2}k_{4}\se_{E_p}\se_{S}+k_{2}k_{3}k_{4}\se_{S}^{2}+k_{2}k_{3}k_{4}\se_{S}\se_{I}+k_{1}k_{2}k_{5}\se_{E_p}+k_{1}k_{3}k_{4}\se_{S}+\\ & \qquad k_{1}k_{3}k_{4}\se_{I}+k_{2}k_{3}k_{5}\se_{S}+k_{1}k_{3}k_{5}, \\
q_2(k,\se)&= k_{1}k_{2}k_{4}\se_{E_p}\se_{S_p}+k_{1}k_{2}k_{4}\se_{E_p}\se_{I}+k_{2}k_{3}k_{4}\se_{S}\se_{S_p}+k_{1}k_{2}k_{5}\se_{E_p}+k_{1}k_{3}k_{4}\se_{S_p}+k_{2}k_{3}k_{5}\se_{S}+k_{1}k_{3}k_{5}.
\end{align*}
By inspecting the signs of these polynomials at positive values of $\se$ and $k$, we conclude that model 1 leads to an increase of the concentration of $I$ while model 2 to a decrease.

In \cite{paradoxical},  apparently paradoxical results on sign-sensitivities were brought up to attention. We recover these phenomena in this work, and encounter a new surprising counter-intuitive result: the concentration of a species $X$ at steady state might decrease as a function of $X$ itself; that is, the concentration of $X$ might decrease after the addition of $X$ to the system   (Section~\ref{hybridhistine kinase:section}).

\medskip
After exemplifying how to find the sensitivities for different types of perturbations (Section~\ref{sec:sensitivities}), we focus on perturbations of concentrations. We argue that, mathematically, it is better posed to discuss responses to a change of an initial concentration rather than to a change of a total amount (Section~\ref{sec:conc}). We then employ recent results relating the sign of \blue{the determinant of the} Jacobian of a function with sign-vectors conditions  from \cite{MullerSigns} to determine, without computing the sensitivities, whether the sign depends on the steady state and parameters of the system (Section~\ref{sec:indep}). We conclude by discussing the existence of multiple steady states and the restriction to stable steady states, which are the only ones observable in an experimental setting (Section~\ref{hybridhistine kinase:section}).

\section{Reaction networks}
The processes we consider are modeled by \textbf{reaction networks}, which can be seen as directed graphs. Specifically, a reaction network consists of a set of species $\{X_1,\dots,X_n\}$, and a directed graph whose nodes are finite linear combinations of species (called \emph{complexes}).  The directed edges are called \emph{reactions}. We let $r$ be the number of reactions. 
An example of a reaction network \cite{SontagSigns}, modeling the transfer of phosphate groups from a kinase $E$ to a substrate $S$ that has two phosphorylation sites is:
 \begin{equation}\label{network:phosphotransfer}
S_0 + E_p \ce{<=>} S_1 + E \qquad S_1 + E_p \ce{<=>} S_2 + E.
\end{equation}
The species of the network are $E, E_p, S_0,S_1,S_2$: $E,E_p$ are the unphosphorylated and phosphorylated forms of the kinase $E$ and $S_0,S_1,S_2$ denote the substrate with no, one or two phosphate groups attached.
The complexes are $S_0 + E_p, S_1 + E,  S_1 + E_p$ and $S_2 + E$. This network, which is used as a running example, originates from the work of Sontag on sign-sensitivities as well \cite{SontagSigns}. See \cite{gunawardena-notes,feinbergnotes} for an expanded introduction to reaction networks. 

The source of a reaction is called the \emph{reactant}, while the target is called the \emph{product}.
We assume the set of species is numbered such that  each complex $y$ can be identified with a vector in $\R^n$; for instance, the complex $X_1+2 X_2$ is identified with the vector $(1,2,0,\dots,0)$ in $\R^n$, where $n$ is the number of species. 
In this way, each reaction $y\rightarrow y'$ gives rise to a vector $y'-y$ in $\R^n$, encoding the net  production of each species with respect to the reaction.  After choosing an order of the set of reactions, these vectors are gathered as columns of a matrix, called the stoichiometric matrix $N\in \R^{n\times r}$.  The stoichiometric matrix for network  \eqref{network:phosphotransfer} is 
\begin{align}\label{eq:N}
 N=\left(\begin{array}{rrrr}
-1&1&0&0\\
1&-1&-1&1\\
0&0&1&-1\\
1&-1&1&-1\\
-1&1&-1&1
\end{array}\right),
\end{align}
where the species are ordered as $S_0,S_1,S_2,E,E_p$.
 
As it is custom within chemical reaction network theory, in this work we model the evolution of the concentration of the species in the network in time by means of a system of ordinary differential equations (ODEs). Specifically, denote by $x_i(t)$ the concentration of $X_i$ at time $t$ (or $x_C(t)$ for a species $C$). 
One chooses the rate $v_{y\rightarrow y'}(x)$ of each reaction of the network, to be a differentiable function from $\R^n_{\geq 0}$ to $\R_{\geq 0}$, and gathers these rates into a vector $v(x)$ from  $\R^n_{\geq 0}$ to $\R^r_{\geq 0}$, using the established orders of the sets of species and reactions. Then 
 the \textbf{evolution equations} of the vector of concentrations $x=(x_1,\dots,x_n)$ takes the form
\begin{equation}\label{eq:ode}
\frac{d x}{dt}= N v(x),\qquad x\in \R^n_{\geq 0}.
\end{equation}
The vector $v(x)$ depends  often on parameters, as it is exemplified below. 
Therefore, we often write $v_k(x)$ to indicate dependence on some vector of parameters $k$, and define
$$ f_k(x)= N v_k(x).$$

Under the assumption of \textbf{mass-action kinetics}, 
we have 
$$ v_{y\rightarrow y'} (x)= k_{y\rightarrow y'} \prod_{i=1}^n x_i^{y_i},$$
with $0^0=1$. Here, $ k_{y\rightarrow y'}>0$ is called the \emph{reaction rate constant} and is treated as a parameter, since it is often unknown. 

Let $B$ be the $n\times r$ matrix such that the $i$-th column is the reactant vector of the $i$-th reaction. 
Then, under the assumption of mass-action kinetics, the right-hand side of system \eqref{eq:ode} can be written equivalently as 
\begin{equation}\label{eq:ode2}  N v_k(x)= N \diag(k) x^B,
\end{equation}
where $x^B\in \R^r_{\geq 0}$ is defined by
$(x^B)_j=  \prod_{i=1}^n x_i^{y_{i}}$ if $y$ is the reactant of the $j$-th reaction.

The reaction rate constant is incorporated in the reaction network as a label of the reactions, such that we write for the network in \eqref{network:phosphotransfer2}
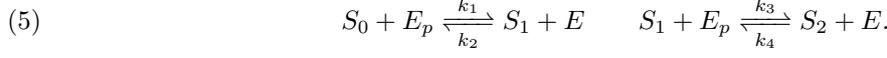
 \begin{equation}\label{network:phosphotransfer2}
S_0 + E_p \ce{<=>[k_1][k_2]} S_1 + E \qquad S_1 + E_p \ce{<=>[k_3][k_4]} S_2 + E.
\end{equation}
We write $x_1,\dots,x_5$ for the concentrations of $S_0,S_1,S_2,E,E_p$ respectively. 
Under the assumption of mass-action kinetics, the matrix $B$ and the vector $v(x)$ are
\begin{equation}\label{eq:B}
B=\left(\begin{array}{cccc}
1&0&0&0\\
0&1&1&0\\
0&0&0&1\\
0&1&0&1    \\
1 & 0 & 1 & 0
\end{array}\right),
\qquad 
v(x)=(k_1x_1x_5, k_2x_2x_4, k_3x_2x_5, k_4x_3x_4 ),
\end{equation}
which together with the matrix $N$ in \eqref{eq:N}, give the following ODE system:
\begin{align*}
\frac{dx_1}{dt} & = - k_{1} x_1 x_5 +k_{2}   x_2 x_4 \\
\frac{dx_2}{dt} & =  k_{1} x_1 x_5 -  k_{2}   x_2 x_4 - k_{3}   x_2 x_5 + k_{4} x_3 x_4
  \\ 
\frac{dx_3}{dt} & =  k_{3}   x_2 x_5 - k_{4} x_3 x_4
\\
\frac{dx_4}{dt} & = k_{1} x_1 x_5 -  k_{2}   x_2 x_4 + k_{3}   x_2 x_5 - k_{4} x_3 x_4
\\
\frac{dx_5}{dt} & = -k_{1} x_1 x_5 +  k_{2}   x_2 x_4 - k_{3}   x_2 x_5 + k_{4} x_3 x_4.
\end{align*}

It is clear from \eqref{eq:ode} that the vector $\frac{dx}{dt}$ belongs to the column span $\im(N)$ of $N$, called the \emph{stoichiometric subspace}. Thus, given an initial condition $x^0$, the solution to \eqref{eq:ode} is confined to the linear subspace $x^0+\im(N)$. Further, both $\R^n_{\geq 0}$ and $\R^n_{>0}$ are also forward-invariant for the trajectories of  \eqref{eq:ode} \cite{Sontag:2001}. 
Each of the sets $(x^0+\im(N)) \cap \R^n_{\geq 0} \subseteq \R^n_{\geq 0}$ is called a \textbf{stoichiometric compatibility class}. In this work we parametrize these sets in two ways. First, choose a matrix $W$ whose rows form a basis of $\im(N)^\perp$. Then, the set $(x^0+\im(N))\cap \R^n_{\geq 0}$ agrees with the subset of $\R^n_{\geq 0}$ defined by the equation
$$W x= W x^0,\qquad x\in \R^n_{\geq 0}.$$
This set is independent of the choice of matrix $W$ and is parametrized by $x^0\in \R^n_{\geq 0}$. 
Let now $d$ be the dimension of $\im(N)^\perp$. Alternatively, one might consider vectors $T=(T_1,\dots,T_d)\in \R^d$ and consider the sets
$$ \big\{ x\in \R^n_{\geq 0} \mid Wx = T \big\}.$$
Each such set corresponds to a stoichiometric compatibility class. However, the set depends both on $W$ and $T$, that is, the vector $T$ alone does not characterize the class. We refer to $T$ as  the vector of \textbf{total amounts} and to $W$ as a matrix of  \textbf{conservation laws}.

\medskip
A matrix of conservation laws $W$ for network \eqref{network:phosphotransfer} is  
\begin{align}\label{W}
W=\left(\begin{array}{ccccc}1&1&1&0&0\\0&1&2&0&1\\0&0&0&1&1\end{array}\right),\end{align}
which gives rise to the following equations for the stoichiometric compatibility class with vector of total amounts $(T_S,T_p,T_E)$
\begin{align}\label{eq:cons1}
x_1+x_2+x_3&=T_S, & x_2+2x_3+x_5&=T_p, & 
x_4+x_5&=T_E.
\end{align}
The first equation encodes that the substrate is conserved, the third that the kinase is conserved, and the second that the phosphate group is either in $S_1, S_2$ or $E_p$, with $S_2$ having two sites. 
Alternatively, we can write the stoichiometric compatibility class of $x^0$ as
\begin{align*}
x_1+x_2+x_3&=x_1^0+x_2^0+x_3^0, & x_2+2x_3+x_5&=x_2^0+2x_3^0+x_5^0, & 
x_4+x_5&=x_4^0+x_5^0.
\end{align*}
We illustrate that it is more meaningful to use this second parametrization when studying sign-sensitivities in Section~\ref{sec:conc}.

In this work, we are interested in the \textbf{positive steady states} of the system \eqref{eq:ode} restricted to a stoichiometric compatibility class. These are defined by a system of equations 
$$  f_k(x)=0, \qquad Wx = Wx^{0}.$$
Since $d$ equations of $f_k(x)$ are redundant (as $WN=0$), we remove them from the system $f_k(x)=0$ and obtain a system with $n$ equations and $n$ variables, which we write as 
\begin{equation}\label{eq:F}
 F_{k,x^0} (x )=0,
\end{equation}
with the first $n-d$ components of $F_{k,x^0}(x)$ obtained from $f_k(x)$ after removing redundant equations, and the last $d$ components are $W (x-x^0)^{tr}=0$. \blue{Here \emph{tr} stands for the transpose of a vector or matrix.} 
If we wish to parametrize stoichiometric compatibility classes with $T\in \R^d$, then we simply write $F_{k,T}(x)=0$ for the corresponding system. 

For network \eqref{network:phosphotransfer}, we remove the equations corresponding to the species   $S_0,S_1,E$, and obtain the following system defining the steady states in the stoichiometric compatibility class of $x^0$:
\begin{align}
k_{3}   x_2 x_5 - k_{4} x_3 x_4 &= 0,\nonumber
\\
  -k_{1} x_1 x_5 +  k_{2}   x_2 x_4 - k_{3}   x_2 x_5 + k_{4} x_3 x_4 &=0,\nonumber \\
  x_1+x_2+x_3- (x_1^0+x_2^0+x_3^0) &=0, \label{eq:F} \\
   x_2+2x_3+x_5-(x_2^0+2x_3^0+x_5^0)&=0, \nonumber\\
x_4+x_5-(x_4^0+x_5^0) &=0. \nonumber
\end{align}

We conclude this section with a definition: we say that a steady state $x^*$ is \textbf{degenerate} if the Jacobian of $F_{k,x^*}(x)$ evaluated at $x^*$ is singular, that is, has vanishing determinant.
This is equivalent to the Jacobian of the function $f_k(x)$ be singular on $\im(N)$, c.f. \cite[Eq (6.1)]{wiuf-feliu}.

\section{Sign-sensitivities}\label{sec:sensitivities}

We consider a reaction network with associated ODE system 
\eqref{eq:ode} and a steady state $\se$. In this section we investigate how the steady state $\se$ changes upon a small perturbation to a parameter of the system, that being either $k$ or $x^0$ (or $T$).
Specifically, we consider the vector of parameters $k\in \R^{m}$ of the rate function $v_k(x)$, and the vector of parameters of initial conditions $x^0\in \R^n$, or the vector of parameters of total amounts $T\in \R^d$. These live in a subspace $\Omega$ of $\R^{M}$ with $M=m+n$ or $M=m+d$, and we write generically $\alpha\in \Omega$ for the vector of parameters of either form $(k,x^0)$ or $(k,T)$. 
We let $\gamma_0\in \Omega$ be the vector of parameters corresponding to our steady state $\se$, such that $F_{\gamma_0}(\se)=0$.

\medskip
\textbf{A formula for sensitivities. }
We consider a continuously differentiable map 
$$\gamma\colon \ (-\epsilon, \epsilon) \rightarrow \Omega,$$
where $\epsilon>0$ and such that $\gamma(0)=\gamma_0$.
If $\se$ is not degenerate, then the Implicit Function Theorem implies that locally around $0$, there is a 
continuously differentiable curve $\se(s)$ with $\se(0)=\se$ and such that $\se(s)$ is a steady state of 
the reaction network and stoichiometric compatibility class with parameters $\gamma(s)$, that is, 
$F_{\gamma(s)} (\se(s))=0$. 

The question we address here is how to determine the sign of the derivative of $\se(s)$ with respect to $s$ at $s=0$, which we denote by $\se'(0)$. We let $\gamma'(s)$ denote the derivative of $\gamma$ with respect to $s$. 

We view $F_{\alpha}(x)$ as a function in both $\alpha$ and $x$ and let  $J_{\alpha,1}(x):=\frac{\partial F_{\alpha}(x)}{\partial x}\in \R^{n\times n}$ denote the Jacobian matrix of $F_{\alpha}(x)$ with respect to the vector $x$ and similarly $J_{\alpha,2}(x):=\frac{\partial F_{\alpha}(x)}{\partial \alpha}\in \R^{n\times M}$ denote the Jacobian matrix of $F_{\alpha}(x)$ with respect to the vector $\alpha$. 
Differentiation of $F_{\gamma(s)} (\se(s))=0$ with respect to $s$ and evaluation at $s=0$ gives
\begin{equation} \label{eq:differentiate}
J_{\gamma_0,1}(\se)  \cdot  \se'(0) + J_{\gamma_0,2}(\se) \cdot\gamma'(0) =0.
\end{equation}
This results in a linear system in $n$ unknowns $\se'_1(0), \dots, \se'_n(0)$ with coefficient matrix  
$J_{\gamma_0,1}(\se)$ and independent term $J_{\gamma_0,2}(\se) \cdot\gamma'(0)$, which we can find if the steady state $\se$ and $\gamma_0$ are given. 
Since the steady state $\se$ is non-degenerate, the coefficient matrix has full rank $n$, and hence this system has a unique solution. Note that neither the coefficient matrix of the linear system $J_{\gamma_0,1}(\se) $  nor  $J_{\gamma_0,2}(\se)$  depend on the specific perturbation $\gamma$. Further, the last $d$ rows of $J_{\gamma_0,1}(\se)$ are $W$.
\blue{When $\alpha=(k,x^0)$, the last $d$ rows of $J_{\gamma_0,2}(\se)$ are $(0\, |\,-W)$, where $0$ is the zero matrix of size $d\times m$. Similarly, when  $\alpha=(k,T)$,  the  last $d$ rows of $J_{\gamma_0,2}(\se)$ are $(0\, | \, -I_{d\times d})$. In both cases  the upper $n-d$ rows are zero in the last $d$ entries.}

\smallskip
Using Cramer's rule, $\se'_i(0)$ is expressed as a fraction where the denominator is the determinant of 
$J_{\gamma_0,1}(\se)$ and the numerator is  the determinant of the matrix obtained by replacing the $i$-th column of $J_{\gamma_0,1}(\se)$ by $- J_{\gamma_0,2}(\se) \cdot\gamma'(0) $.
When the rate functions are mass-action, then $\se'_i(0)$ becomes a rational function in the parameters and the entries of $\se=(\se_1,\dots,\se_n)$. 

\medskip
In the general scenario, nor the steady state $\se$ nor the parameter value $\gamma_0$ are known, and therefore
we aim at determining the sign of $\se'_i(0)$ for all values of $\se$ and $\gamma_0$, and at deciding whether this sign is independent of these values. 
As it has been used in several works, e.g. \cite{FeliuPlos,Dickenstein:structured}, the set of all positive steady states is studied by means of a \textbf{parametrization} 
$$ \varphi\colon U \rightarrow \R^n_{>0},$$
such that the image of $\varphi$ is the set of positive steady states  (see \cite{FeliuPlos} for strategies to find parametrizations).

\medskip
We illustrate this framework and   computations with selected perturbations $\gamma$ for our running example \eqref{network:phosphotransfer}.
First, note that due to the matrix of conservation laws in \eqref{W}, the steady state equations for $x_1$, $x_2$ and $x_4$ are redundant. Thus, a positive steady state is simply a point $x\in \R^5_{>0}$ satisfying the steady state equations for $x_3$ and $x_5$:
\begin{align*}
0 & =  k_{3}   x_2 x_5 - k_{4} x_3 x_4
&
0 & = -k_{1} x_1 x_5 +  k_{2}   x_2 x_4 - k_{3}   x_2 x_5 + k_{4} x_3 x_4,
\end{align*}
or equivalently
\begin{align}\label{eq:ss}
0 & =  k_{3}   x_2 x_5 - k_{4} x_3 x_4
&
0 & = -k_{1} x_1 x_5 +  k_{2}   x_2 x_4. 
\end{align}
Any solution to this system is of the form
\begin{align}\label{parametirzation:example}
\se= \Big( x_1, x_2, \frac{k_{2}k_{3}x_2^2}{k_{1}k_{4}x_1},  x_4, \frac{k_{2}x_4x_2}{k_{1}x_1}\Big),
\end{align}
that is, the set of positive steady states is parametrized by $x_1$, $x_2$ and $x_4$, where $U=\R^3_{>0}$. If these three variables are positive, then so is $\se$. 
Let $\alpha=(k_1,k_2,k_3,k_4, x_1^0,x_2^0,x_3^0,x_4^0,x_5^0) \in \R^9_{>0}$ be the vector of parameters. The function $F_{\alpha}(x)$ is taken to be the left-hand side of the system in \eqref{eq:F}. This leads to the following Jacobian matrices:
\begin{align}
\label{JW}
J_{\alpha,1}(x) &=\left(\begin{matrix}0 & k_{3} x_5& - k_{4} x_4 & - k_{4} x_3 & k_{3} x_2 \\
       - k_{1} x_5 & k_{2} x_4 - k_{3} x_5 & k_{4} x_4 & k_{2} x_2 + k_{4} x_3 & - k_{1} x_1 - k_{3} x_ 2 \\
       1 & 1 & 1 & 0 & 0\\
       0 & 1 & 2 & 0 & 1\\
       0 & 0 & 0 & 1 & 1 \end{matrix}\right), \\
 J_{\alpha,2}(x)  &=  \begin{pmatrix} 0&0&x_{{2}}x_{{5}}&-x_{{3}}x_{{4}}&0
&0&0&0&0\\  -x_{{1}}x_{{5}}&x_{{2}}x_{{4}}&-x_{{2}}x_{{5}}&x_{{3}}x_{{4}}&0&0&0&0&0\\ 
0&0&0&0&-1&-1&-1&0&0 \\ 0&0&0&0&0&-1&-2&0&-1\\ 0&0&0&0&0 &0&0&-1&-1\end{pmatrix}. \label{JW2}
\end{align}

\medskip
\textbf{Perturbing $k_1$. } We consider first  the perturbation that maps $k_1$ to $k_1+s$. This gives $\gamma'(0)=(1,0,0,0,0,0,0,0,0)$ and hence $J_{\alpha,2}(x) \cdot\gamma'(0)$ is simply the first column of $J_{\alpha,2}(x)$ in \eqref{JW2}, which is  $(0,-x_1 x_5,0,0,0)^{tr}$. 
We solve system \eqref{eq:differentiate} with these data and obtain for $\se=(x_1,\dots,x_5)$ that 
\begin{align*}
\se'_1(0) &=\frac {-x_{{1}}x_{{5}} \left( k_{{3}}x_{{2}}+k_{{3}}x_{{5}}+k_{{4}}x_{{3}}+k_{{4}}x_{{4}} \right) }{q(k,x)}, & \se'_3(0) & =\frac {x_{{1}}x_{{5}} \left( -k_{{3}}x_{{2}}+k_{{3}}x_{{5}}-k_{{4}}x_{{3}} \right) }{q(k,x)},\\ 
\se'_2(0) & =\frac {x_{{1}}x_{{5}} \left( 2\,k_{{3}}x_{{2}}+2\,k_{{4}}x_{{3}}+k_{{4}}x_{{4}} \right) }{q(k,x)},&
\se'_4(0) & =\frac {x_{{1}}x_{{5}} \left( 2\,k_{{3}}x_{{5}}+k_{{4}}x_{{4}} \right) }{q(k,x)},\\
\se'_5(0) & =\frac {-x_{{1}}x_{{5}} \left( 2\,k_{{3}}x_{{5}}+k_{{4}}x_{{4}} \right) }{q(k,x)},
\end{align*}
where 
\begin{multline*}
q(k,x)=2\,k_{1}k_{3}x_{1}x_{5}+k_{1}k_{3}x_{2}x_{5}+k_{1}k_{3}x_{5}^{2}+k_{1}k_{4}x_{1}x_{4}+k_{1}k_{4}x_{3}x_{5}+k_{1}k_{4}x_{4}x_{5}\\ +2\,k_{2}k_{3}x_{2}x_{4}+2\,k_{2}k_{3}x_{2}x_{5}+k_{2}k_{4}x_{2}x_{4}+2\,k_{2}k_{4}x_{3}x_{4}+k_{2}k_{4}x_{4}^{2}.
\end{multline*}
We readily see that $\se_1$ and $\se_5$ \blue{decrease}, and $\se_2$ and $\se_4$ \blue{increase} when $k_1$ is slightly increased. The sign of $\se_3'(0)$ is not determined (yet). But we have not imposed that $\se$ is a steady state. In order to do that, we evaluate  $\se_3'(0)$ in the parametrization and obtain that the sign of $\se_3'(0)$ at a steady state is the sign of
$$ -k_3x_2+k_3\frac{k_{2}x_4x_2}{k_{1}x_1}-k_ 4\frac{k_{2}k_{3}x_2^2}{k_{1}k_{4}x_1}   =
\frac{k_{3} x_2}{k_1x_1} \big(  -k_1x_1+k_{2}x_4 -  k_{2}x_2\big).$$
Clearly, this expression can be positive, negative or zero, after  appropriately choosing $k_1,k_2,x_1,x_2,x_4$. We conclude that the sign of the change of $\se_3$ with respect to this perturbation is not parameter and variable independent, and therefore information on the specific value of the steady state is required.  
 
 \medskip
\textbf{Perturbing $x_4^0$. }
 We perturb now  $x_4^0$ (that is, the initial concentration of $E$) by the addition of a small amount $s$. We now have $\gamma'(0)=(0,0,0,0,0,0,0,1,0)$ and hence $J_{\alpha,2}(x) \cdot\gamma'(0)$ is the eighth column of $J_{\alpha,2}(x)$, which is  $(0,0,0,0,-1)^{tr}$. 
We solve the resulting system \eqref{eq:differentiate} and obtain for $\se=(x_1,\dots,x_n)$ that 
\begin{align*}
\se'_1(0)&= \frac{-k_1k_4x_1x_3+k_2k_3x_2^2+k_2k_3x_2x_5+k_2k_4x_2x_4+k_2k_4x_3x_4}{q(k,x)}.
\end{align*}
After evaluating the numerator of $\se'_1(0)$ in the parametrization, we obtain
$$ \frac{2k_2^2 k_3x_2^2x_4}{k_1x_1} + k_2k_4x_2x_4,$$
which only attains positive values. We conclude that the concentration of $S_0$ at steady state increases when an infinitesimal amount of $E$ is added to the system.

\medskip
We proceed in the same way to determine $\se'_i(0)$ after perturbing each of the reaction rate constants $k_j$ and initial concentrations $x_j^0$ one by one by adding a small amount $s$. The sign-sensitivities are summarised in Table~\ref{table:signs}. 
A seemingly striking insight of this table is that an increase of a certain species can be paired with both an increase or decrease of another species, depending on the perturbation applied.
For instance,  $E_p$ increases ($\se_5'(0)>0)$ after increasing either $x_3^0$ or $x_4^0$, while $E$ decreases ($\se_4'(0)<0$)  for the first perturbation and increases  ($\se_4'(0)>0$)   for the second. This highlights that perturbation studies need to be appropriately interpreted, as it would be wrong to conclude, out of the column for $x_3^0$, that $E$ decreases when $E_p$ increases. That is, one needs to pair the direct perturbation to the response, and not two responses to a perturbation. 
This ``paradoxical'' result has been first pointed out in \cite{paradoxical}. 

\begin{table}[t]
\begin{center}
\begin{tabular}
{c||c|c|c|c|c|c|c|c|c|c|}
 & $k_1$ & $k_2 $ & $k_3$ & $k_4$ &  $x_1^0$ & $x_2^0$ & $x_3^0$ & $x_4^0$ & $x_5^0$ \\\hline\hline
$\se'_1(0)$ &  $-$  &  $+$  & \cellcolor{gray!30!white}   $+\tau_1$ &  \cellcolor{gray!30!white}$-\tau_1$  &$+$&$+$& \cellcolor{gray!30!white} $-\tau_1$& \cellcolor{gray!30!white} $+$& \cellcolor{gray!30!white} $-$\\ \hline
$\se'_2(0)$ &   $+$ &  $-$ &  $- $ & $+$  &$+$&$+$&$+$& \cellcolor{gray!30!white} $-\tau_2$&\cellcolor{gray!30!white}  $+\tau_2$ \\ \hline
$\se'_3(0)$ & \cellcolor{gray!30!white}$-\tau_1$  & \cellcolor{gray!30!white} $+\tau_1$ & $+$  & $-$ &  \cellcolor{gray!30!white} $-\tau_1$&$+$&$+$& \cellcolor{gray!30!white} $- $& \cellcolor{gray!30!white} $+$ \\ \hline
$\se'_4(0)$ &  $+$ &    $-$   & $+$  & $-$ & $+ $ & \cellcolor{gray!30!white}  $-\tau_2$ & $-$ &  $+$  &$+$\\ \hline$\se'_5(0)$ &   $-$ &   $+$ & $-$ & $+$ & $-$ & \cellcolor{gray!30!white}$+\tau_2$&$+$&$+$&$+$\\ \hline
\end{tabular}
\end{center}
\caption{Sign-sensitivities with respect to adding a small amount to each of the parameters. Each column gives the sign-sensitivity with respect to one parameter. $\tau_1$ is the sign of $k_1x_1 + k_2(x_2-x_4)$ and $\tau_2$ is the sign of $k_1k_4x_1^2-k_2k_3x_2^2$ at the steady state (which can be zero). Gray cells are determined using the parametrization, and for the other cells the sign is determined   for all $x\in \R^5_{>0}$.}\label{table:signs}
\end{table}

 \medskip
 \textbf{General perturbations. }  The outlined framework accommodates all types of perturbations, not only consisting in adding a small amount to one of the parameters. We illustrate this with two perturbations:  in the first   we scale two reaction rate constants by $s$, and in the second a small amount $s$ is added to $x_4^0$ and $x_5^0$ simultaneously. 

\smallskip
First, consider the perturbation to $\alpha=(k_1,k_2,k_3,k_4, x_1^0,x_2^0,x_3^0,x_4^0,x_5^0)$
such that
$$\gamma(s)= (sk_1,k_2,sk_3,k_4, x_1^0,x_2^0,x_3^0,x_4^0,x_5^0).$$
Then $\gamma'(0)=(k_1,0,k_3,0,0,0,0,0,0)$ and $J_{\gamma_0,2}(\se)\cdot \gamma'(0)$  is the vector $(k_3x_2x_5,-k_1x_1x_5-k_3x_2x_5)^{tr}$. Solving the corresponding system \eqref{eq:differentiate}, we obtain that   $\se_1'(0)$ and $\se_5'(0)$ are negative, $\se_3'(0)$ and $\se_4'(0)$ are positive, and $\se_2'(0)$ can be of either sign. Here only the signs of $\se_4'(0)$ and $\se_5'(0)$ can be determined without the parametrization. 

If instead we consider the perturbation $\gamma(s)= (sk_1,sk_2,k_3,k_4, x_1^0,x_2^0,x_3^0,x_4^0,x_5^0)$, then all derivatives become zero, that is, the steady state is invariant under simultaneously scaling $k_1$ and $k_2$ (as it is readily seen from \eqref{eq:ss}).

\smallskip
Consider next the perturbation 
 to  $\alpha$ such that 
 $$\gamma(s)=(k_1,k_2,k_3,k_4, x_1^0,x_2^0,x_3^0,x_4^0+s,x_5^0+s).$$
 Then
 $\gamma'(0)=(0,0,0,0,0,0,0,1,1)$
and $J_{\gamma_0,2}(\se)\cdot \gamma'(0)$ is the sum of the last two columns of $J_{\gamma_0,2}(\se)$, namely the vector
$ (0,0,0,-1,-2)^{tr}$. 
Then the solution $\se_i'(0)$ to the corresponding system is the sum of $\se_i'(0)$ for the perturbation $x_4^0 \mapsto x_4^0 +s$ and $\se_i'(0)$ for the perturbation $x_5^0 \mapsto x_5^0 +s$. By Table~\ref{table:signs}, the sign of $\se_4'(0)$ and  $\se_5'(0)$ is $+$. We further obtain that the sign  of $\se_1'(0)$, $\se_2'(0)$ and  $\se_3'(0)$ can be any  of $-,0,+$.

\medskip
In this example it is straightforward to decide whether the numerator of $\se_i'(0)$ can attain any sign when the polynomial has coefficients of both signs. For larger systems, an often successful approach consists on investigating the vertices of the Newton polytope associated with the polynomial. If two of the vertices correspond to monomials with coefficients of opposite signs, then the polynomial attains all signs for positive values of the variables. This strategy has been used in numerous recent works with chemical reaction network theory, e.g.  \cite{FeliuPlos,obatake:hopf,conradi:mixed}, and we refer the reader to \cite{FeliuPlos} for an expository account.

\begin{remark}
In \cite{brehm:sensitivity} the authors provide structural conditions to determine whether a sign-sensitivity is zero, for ODE systems arising from a subclass of rate functions that does not include mass-action. In that work, perturbations on reaction rate constants of the form $k_i\mapsto k_i + s$ are considered using the corresponding equation \eqref{eq:differentiate}. 
In Metabolic Control Analysis (MCA) \cite{Fell:MCA}, so called \emph{flux/concentration control coefficients} are considered. The latter measures sensitivity similarly to here, as the derivative of the logarithm of a concentration $\se_i$ with respect to the logarithm of another concentration $x_j^0$ is taken, or what is equivalent
\[ \se_i'(0) \cdot \tfrac{x_j^0}{\se_i(0)}. \]
These are found using \eqref{eq:differentiate} as well, after adjusting the formula. Since $\tfrac{x_j^0}{\se_i(0)}$ is positive, this factor is redundant when considering sign-sensitivities, but in MCA, of relevance is the value of this (normalized) derivative, and not only its sign.  See \cite{gunawardena:MCA} for a gentle introduction to MCA and control coefficients. 
\end{remark}

\section{Perturbing concentrations}\label{sec:conc}
In this section we take a closer look at perturbations caused by a change in the stoichiometric compatibility class. 
The first observation we make is that perturbations of the total amounts might lead to apparently contradictory results. To see this, consider network \eqref{network:phosphotransfer}, 
with conservation laws and total amounts as given in \eqref{eq:cons1} and corresponding function $F_{k,T}(x)$. The vector of parameters is now $\alpha=(k_1,k_2,k_3,k_4,T_S,T_p,T_E)$. Under the perturbation $\gamma$ on the total amount of phosphorylated proteins $T_p \mapsto T_p + s$, we have $\gamma'(0)=(0,0,0,0,0,1,0)$ and 
we obtain
$$ \sign(\se_1'(0))= - ,\quad \sign(\se_2'(0))=\pm, \quad \sign(\se_3'(0))=+,\quad \sign(\se_4'(0))= -,\quad \sign(\se_5'(0))= -.  $$
We consider now another matrix of conservation laws $W'$, with  same second row as $W$ in \eqref{W}:
$$W'=\left(\begin{array}{ccccc}1&0&-1&0&-1\\0&1&2&0&1\\0&0&0&1&1\end{array}\right). $$
We perform the same perturbation on $T_p$, and obtain the following sign-sensitivities:
$$ \sign(\se_1'(0))= + ,\quad \sign(\se_2'(0))=+, \quad \sign(\se_3'(0))=+,\quad \sign(\se_4'(0))= \pm,\quad \sign(\se_5'(0))= \pm.  $$
Although we did not change the expression for the total amount $T_p$, the sign-sensitivities changed drastically. For example, $S_0$ decreases when $T_p$ is increased for the first matrix of conservation laws, while it decreases for the second choice. 
We conclude that perturbations with respect to total amounts might not be meaningful and need to be appropriately interpreted \blue{as perturbations of the considered system}.

\medskip
We proceed to investigate perturbations with respect to initial concentrations, and show that in this case, the sign-sensitivities do not depend on the choice of matrix of conservation laws. For the rest of the section we let $\alpha=(k,x^0)$. Note that $J_{\alpha,2}(x)$ is independent of $x^0$ since $F_\alpha(x)$ is linear in $x^0$. 

\begin{lemma}
Consider the perturbation $\gamma$ sending $x_i^0$ to $x_i^0+s$, and being the identity on the other parameters. 
For $j=1,\dots,n$, the derivative $\se_j'(0)$ does not depend on the basis of $\im(N)^\perp$ used to construct the function $F_\alpha(x)$.
\end{lemma}
\begin{proof}
Let $W,W'$ be two matrices of conservation laws. Then there exists an invertible $d\times d$ matrix $A$ such that $W'=AW$. 
Let $F_{\alpha}(x)$ and $F'_{\alpha}(x)$ be the corresponding steady state functions from \eqref{eq:F}, and denote by $J,J'$ (with the appropriate subindices) their Jacobian matrices respectively. 
Then
$$ J'_{\gamma_0,1}(\se)  \cdot  \se'(0) + J'_{\gamma_0,2}(\se) \cdot\gamma'(0)
=  \left(\begin{array}{cc}I_{n-d}& 0\\ 0&A\end{array}\right) \Big(J_{\gamma_0,1}(\se)  \cdot  \se'(0) + J_{\gamma_0,2}(\se) \cdot\gamma'(0)\Big),
$$
where $I_{n-d}$ is the identity matrix of size $n-d$ (see text after \eqref{eq:differentiate}). 
Since $ \left(\begin{array}{cc}I_{n-d}&0\\  0&A\end{array}\right)$ is invertible,  the solution to
\eqref{eq:differentiate}  for $W$ and $W'$ agree.
 \end{proof}

Having established that $\se_j'(0)$ does not depend on the choice of matrix of conservation laws, we can easily prove a series of lemmas by appropriately selecting this matrix. \blue{First, we note that if a concentration does not take part of any conservation law, then all sensitivities with respect to changes to this concentration are zero. }

\begin{lemma}\label{lem:zero}
Consider the perturbation   sending $x_i^0$ to $x_i^0+s$, and being the identity on the other parameters.
If the $i$-th component of all vectors in $\im(N)^\perp$ is zero, then $\se_j'(0)=0$ for all $j=1,\dots,n$. In other words, the sign-sensitivities are all zero. 
\end{lemma}
\begin{proof}
By hypothesis, the $i$-th column of any matrix of conservation laws is zero. Consequently $J_{\gamma_0,2}(\se) \cdot\gamma'(0)$ is the zero vector and the only solution to \eqref{eq:differentiate} is the zero vector.
\end{proof}

In the next proposition we discuss perturbations with respect to an initial concentration that only appears in one conservation law, and how it relates to the  perturbation with respect to the corresponding total amount.

\begin{proposition}\label{prop:form}
Assume the matrix of conservation laws $W=(w_{j,i})$ is such that the $i$-th column has only one non-zero entry, that is, there exists $\ell$ such that 
\begin{align*}
w_{\ell',i} & =0 \quad \textrm{for }\ell'\neq \ell \quad\textrm{and}\quad w_{\ell,i}  \neq 0.
\end{align*}
Let $M_j$  be the minor of $ J_{\gamma_0,1}(\se)$ obtained by removing the $j$-th column and the $(n-d+\ell)$-th row, divided by  $\det J_{\gamma_0,1}(\se)$. 

Then  
\begin{itemize}
\item $\se_j'(0)$ for the perturbation $\gamma_i$ sending $x_i^0$ to $x_i^0+s$
equals $(-1)^{n-d+\ell+j}w_{\ell,i} M_j$.
\item $\se_j'(0)$ for the perturbation $\gamma^*_\ell$ sending $T_\ell$ to $T_\ell+s$
equals $(-1)^{n-d+\ell+j} M_j$.
\end{itemize}
\end{proposition}
\begin{proof}
The statement of the proposition follows after noticing that $J_{\gamma_0,2}(\se) \cdot\gamma_i'(0)$ is the  vector with $-w_{\ell,i}$ in the $(n-d+\ell)$-th entry and zero everywhere else, and the vector $J_{\gamma_0,2}(\se) \cdot (\gamma^*_\ell)'(0)$  is $-1$ in the  $(n-d+\ell)$-th entry and zero otherwise.
\end{proof}

In particular, it follows from  Proposition~\ref{prop:form} that if $w_{\ell,i}>0$, then the perturbations with respect to $x_i^0$ or $T_\ell$ yield sensitivities with the same sign.  As a consequence, if the matrix $W$ is row reduced, then perturbation with respect to a slight increase of a total amount can be interpreted as the perturbation with respect to $x_i^0$ for $i$ the index of the first non-zero entry of the corresponding conservation law. Hence, the value of the perturbation with respect to this total amount is well defined under the restriction that the other conservation laws do not involve $x_i$.

An immediate consequence of  Proposition~\ref{prop:form} is that  \blue{if two columns $i,i'$ of $W$ have both only one non-zero entry,  at the $\ell$-th row, then sensitivities with respect to $x_i^0$ and $x_{i'}^0$ agree up to the product of these non-zero entries.} This is summarized in the following lemma. 

\begin{lemma}\label{corollary}
If there exists a basis $\{w_1,\dots,w_d\}$ of $\im(N)^\perp$ and three indices $i,i',\ell$ such that 
\begin{align*}
w_{\ell',i} & =w_{\ell',i'}=0 \quad \textrm{for }\ell'\neq \ell & 
w_{\ell,i} & \neq 0, & w_{\ell,i'} &\neq 0,
\end{align*}
then  $c_j'(0)$ 
 for the perturbation $\gamma_i$  sending $x_i^0$ to $x_i^0+s$ agrees with that for the perturbation $\gamma_{i'}$  sending $x_{i'}^0$ to $x_{i'}^0+s$ times $w_{\ell,i'}/w_{\ell,i}$. 
\end{lemma}
\begin{proof}
By Proposition~\ref{prop:form}, the numerator of  $c_j'(0)$ for the perturbation $\gamma_i$ is  $(-1)^{n-d+\ell + j}w_{\ell,i}$ times the minor of $ J_{\gamma_0,1}(\se)$ obtained by removing the $j$-th column and the $(n-d+\ell)$-th row, and for $\gamma_{i'}$, this same minor is multiplied by $(-1)^{n-d+\ell+j} w_{\ell,i'}$. 
\end{proof}

As an example, consider the two rivaling models in the introduction. For model 1, a matrix of conservation laws is 
$$W= \begin{pmatrix}1 & 1 & 0 & 0 & 0 & 0\\
0 & 0 & 1 & 1 & 0 & 1\\
0 & 0 & 0 & 0 & 1 & 1\end{pmatrix},$$
with the order of species $E, E_p, S_0,S_1,I,Y$. The pairs $(E,E_p)$ and $(S_0,S_1)$ satisfy the hypothesis of Lemma~\ref{corollary}. Therefore, for any $j$, the sign of $\se_j'(0)$ is the same when either $E$ or $E_p$ are increased, and similarly, is the same when either $S_0$ or $S_1$ are increased.

\section{Parameter-independent sign-sensitivities}\label{sec:indep}
Although the computation of $\se_i'(0)$ by solving system \eqref{eq:differentiate} might seem straightforward, it requires the computation and analysis of two symbolic determinants. As noticed in  \cite{feliu-bioinfo,baudier:biomodels}, these computations are expensive for relatively large reaction networks, as those encountered in applications. 
In this section we investigate an alternative approach to decide whether the sign of $\se_i'(0)$ does not depend on the value of the parameters nor on $\se = (x_1,\dots,x_n)\in \R^n_{>0}$, that is, without imposing that $\se$ is a steady state. We do this for perturbations of an initial concentration as in the previous section. Once it has been established that the sign is independent of the parameters and $x$, then it can easily be determined after arbitrarily choosing values.

The subsequent results are based on the study of \emph{injective networks} and sign vectors, as presented in \cite{MullerSigns}. 
For that, some notation needs to be introduced. The sign-vector $\sigma(v)$ of a vector $v$ is obtained by taking the sign component-wise. For $V\subseteq \R^n$, let $\sigma(V)$ be the set of  sign vectors of all elements in $V$, and $\Sigma(V)$ be the subset of $\R^n$ containing all, possibly lower dimensional, orthants that $V$ intersects. 
Consider a function in $\R^n_{>0}$ of the form 
$$ g_k(x)=N \diag(k) x^B,$$
with $N\in \R^{n\times r}$ of rank $n-d$, $B\in \R^{r\times n}$, $k\in \R^r_{>0}$. Let $S\subseteq \R^n$ be a vector subspace of dimension $d$ and  $N'$ be a submatrix of $N$ given by $n-d$ linearly independent rows of $N$ (such that $\ker(N)=\ker(N')$). 
Define a function $G_k(x)$ with the first $n-d$ components equal to $N'\diag(k) x^B$, and the last $d$ components  $Wx^{tr}$, for $W$ any basis of $S^\perp$. Then the determinant of the Jacobian of $G_k$ is a polynomial in $k$ and $x$ such that all coefficients have the same sign if and only if 
\begin{equation}\label{eq:sign}
\sigma(\ker(N))\cap\sigma\big(B^{tr}(\Sigma(S \backslash\{0\}))\big)=\emptyset.
\end{equation}
(see \cite{MullerSigns}).  Further, if any of these conditions hold, the function $G_k(x)$ is injective on all cosets $(x^0+S)\cap \R^n_{>0}$ for any choice of $k$. 
\blue{To understand how \eqref{eq:sign} arises, one first notes that the determinant of the Jacobian $J_{G_k}$ of $G_k$ is a polynomial in $k$ and $x$ such that all coefficients have the same sign if and only if it never vanishes. Vanishing of $\det  J_{G_k}$  means that the Jacobian of $G_k$ has non-trivial kernel, that is, there exists a non-zero vector $u$ in $\ker (J_{g_k})$ which further satisfies  $Wu=0$, i.e. $u\in S\setminus \{0\}$. Using $J_{g_k}(x)=N \diag(k) B^{tr} \diag(\tfrac{1}{x})$, we have $u\in \ker(J_{g_k}(x))$ if and only if $\ker(N)$ contains $\diag(k) B^{tr} \diag(\tfrac{1}{x})u.$ Condition \eqref{eq:sign} arises from noticing that varying $x$ and $u\in S\setminus \{0\}$ means considering all orthants that $S\setminus \{0\}$ intersects, that is,  $\Sigma(S \backslash\{0\})$, and then a vector $k$ such that $\diag(k) B^{tr} \diag(\tfrac{1}{x})u$ belongs to $\ker(N)$ exists if and only if $B^{tr} \diag(\tfrac{1}{x})u$ has the sign of some vector in $\ker(N)$. For details of this construction we refer the reader to \cite{MullerSigns}. }

When \blue{applying \eqref{eq:sign}} to a reaction network with mass-action kinetics, we consider $S=\im(N)$, and if  \eqref{eq:sign} holds, then the reaction network is said to be injective.  By verifying the sign equality \eqref{eq:sign} with $N$ the stoichiometric matrix, $S=\im(N)$, and $B$ the exponent matrix in \eqref{eq:ode2}, we can determine whether the sign of  the determinant $J_{\alpha,1}(x)$, and hence of the denominator of $\se_i'(0)$, is constant. We emphasize that we do not impose that $x$ is a steady state in the computations in this section. 

In order to study the numerator, \blue{we interpret it as the Jacobian of a function of the form $g_k(x)$ as above, with the  same coefficient matrix $N$, and suitable exponent matrix $B_j$ and vector space $S_j$. Afterwards we apply \eqref{eq:sign}. The specific form of these objects is given in the next proposition. }

\begin{proposition}\label{prop:signs} Assume mass-action kinetics. 
Consider the perturbation $\gamma_i$ and assume that $\im(N)^\perp$ contains vectors with non-zero $i$-th entry. Let 
$W$ be a matrix of conservation laws  such that the only row where the $i$-th component is non-zero is the first, where it takes the value $1$. Let $S_{j}\subseteq\R^{n-1}$ be the kernel of the vector subspace spanned by the rows of the matrix obtained from $W$, by deleting the first row and the $j$-th column, and let $B_j$ be obtained from $B$ by removing the $j$-th row.

Then the sign of the numerator of $\se_j'(0)$, as a function of $k$ and $x$, is independent of $k$ and $x$ if and only if 
\[\sigma(\ker(N))\cap\sigma(B_j^{tr}(\Sigma(S_{j}\backslash\{0\})))=\emptyset.\]
\end{proposition}
\begin{proof}
Let $\widehat{x}=(x_1,\dots,x_{j-1},x_{j+1},\dots x_n)\in \R^{n-1}$ be the vector $x$ where the $j$-th entry is deleted.
Let $N'$ be a matrix formed by $n-d$ linearly independent rows of $N$ (such that $\ker(N)=\ker(N')$), and for $\ell=1,\dots,r$, let 
$\widehat{k}_\ell = k_\ell x_j^{y_j}$ if the reactant of the $\ell$-th reaction is $y$. 

With the choice of $W$ and the considerations before Lemma~\ref{corollary}, the numerator of $\se_j'(0)$ is, up to a constant sign, the determinant of the submatrix $J'$ of $ J_{\alpha,1}(\se)$ obtained by removing the $j$-th column and the $(n-d+1)$-th row. 
This matrix $J'$ agrees with the Jacobian of the function $G_k(\widehat{x})$ in $\R^{n-1}$ with the first $n-d$ entries equal to $N' \diag\big(\, \widehat{k}\, \big) \widehat{x}^{B_j}$ and bottom $d-1$ entries $W' \widehat{x}^{tr}$, where $W'$ is the matrix obtained from $W$, by deleting the first row and the $j$-th column.

As recalled in \eqref{eq:sign}, by \cite{MullerSigns}, the sign of the determinant of $J'$ does not depend on $\widehat{k}$ nor $\widehat{x}$, hence on $k$ nor $x$, if and only if the sign condition in the statement holds. 
\end{proof}

To illustrate this result, we consider network \eqref{network:phosphotransfer}, the perturbation of $x_2^0$ by adding $s$, and focus on $\se_1'(0)$. By Table~\ref{table:signs}, we already know that the sign of $\se_1'(0)$ is $+$, and this holds for any $x$ without imposing the steady state condition. In particular, in the notation of Proposition~\ref{prop:signs}, $i=2$ and $j=1$.

The kernel of $N$ in \eqref{eq:N} is generated by the vectors $(1,1,0,0)$ and $(0,0,1,1)$, and hence for any $u=(u_1,u_2,u_3,u_4)$ in $\ker(N)$, the sign of $u_1$ and $u_2$ agree, and the sign of $u_3$ and $u_4$ agree.
 The matrix $B_1$  obtained by removing the first row of $B$ in \eqref{eq:B} and a matrix of conservation laws satisfying the hypothesis of Proposition~\ref{prop:signs} are given as
 $$
B_1=\left(\begin{array}{cccc}
0&1&1&0\\
0&0&0&1\\
0&1&0&1    \\
1 & 0 & 1 & 0
\end{array}\right),\qquad
W=\left(\begin{array}{ccccc} 1&1&1&0&0\\
-1&0&1&0&1\\
0&0&0&1&1\end{array}\right).$$
Hence $S_1=\ker\left(\begin{array}{ccccc}  
 0&1&0&1\\
 0&0&1&1\end{array}\right)$ is generated by 
$(1,0,0,0)$ and $(0,-1,-1,1)$. In particular for any vector $(a,b,c,d)$ in $\Sigma(S_{1}\backslash\{0\})$,   at least one entry is non-zero, the signs of $b$ and $c$ agree and are opposite to the sign of $d$, unless $b,c,d$ are zero. 
Now, $B_1^{tr}$ times $(a,b,c,d)^{tr}$ is the vector
 $u=(d,a+c,a+d,b+c)^{tr}$. If $d$ is positive, then $b$, $c$ and $b+c$ are negative, and for $u$ to have the sign of a vector in $\ker(N)$, it is necessary that $a+d$ is negative (hence $a$ negative) and $a+c$ is positive (hence $a$ positive), a contradiction. Similarly, we argue that if $d$ is negative, then $u$ does not have the sign of any vector in $\ker(N)$. Finally, if $d$ is zero, then so are $b,c$, and hence $u=(0,a,a,0)$, which has the sign of a vector in $\ker(N)$ only if $a=0$, a contradiction. 
 
 We have therefore verified that the sign condition in Proposition~\ref{prop:signs} holds, and therefore $\se_1'(0)$ does not depend neither on $k$ nor on $x$.
  Clearly, finding the sign vectors manually is not optimal at all. In \cite{MullerSigns}, see also \cite{dickenstein:messi}, strategies to verify whether this equality holds are presented.

\section{Stable vs unstable steady states}\label{hybridhistine kinase:section}
In the previous sections we have not taken into consideration whether the steady states are asymptotically stable or unstable. In practice, in an experimental setting, only stable steady states are observable. Although it is often not possible to restrict parametrizations of the set of steady states to only stable steady states, some relevant information can be extracted from the sign of the determinant of the Jacobian $J_{\alpha,1}(x)$. 

Specifically, assume the function $F_{\alpha}(x)$ is constructed from a matrix of conservation laws $W$ that is row reduced, and let $i_1,\dots,i_d$ be the indices of the first non-zero entries of the rows of $W$.  The first $n-d$ entries of $F_\alpha(x)$ can be chosen to be the entries of $f_k(x)$ with index different from $i_1,\dots,i_d$. 
Let $\tau=\sum_{\ell=1}^d (n-\ell  - i_\ell)$ 
be the sign of the permutation that reorders the entries of $F_\alpha(x)$ such that the entries defined by $W$ are at entries $i_1,\dots,i_d$. 
Then, by \cite[Prop. 5.3]{wiuf-feliu},   the determinant of $J_{\alpha,1}(x)$ is $(-1)^\tau$ times the product of the $n-d$ nonzero eigenvalues of the Jacobian of $f_k$ evaluated at $x$. Hence, if the steady state is hyperbolic and asymptotically stable, then necessarily the sign of this determinant  is $(-1)^{\tau + n-d}$.

In the previous examples, the determinant of the Jacobian $J_{\alpha,1}(x)$ at a steady state had a constant sign, which was actually $(-1)^{\tau+n-d}$, and hence in accordance with stability. We now illustrate by means of an example what can be said about sensitivities when the network has unstable steady states. 

For that, we consider a simple model of a hybrid histidine kinase from \cite{feliu:unlimited}. The network is depicted in Figure~\ref{fig:HK}(a). Under mass-action kinetics, there exist stoichiometric compatibility classes for which this network has three positive steady states  \cite{feliu:unlimited}, two of which are asymptotically stable \cite{torres:stability}. 
Further, by \cite{FeliuPlos}, the network admits three positive steady states in some stoichiometric compatibility class if and only if $k_3>k_1$. If $k_1\geq k_3$, then the network has exactly one positive steady state in each stoichiometric compatibility class.

\begin{figure}[t]
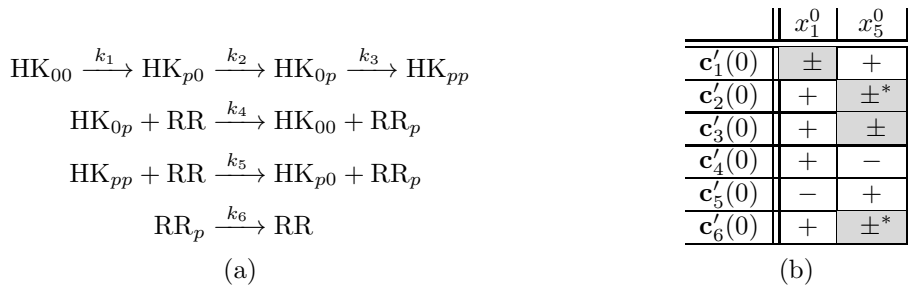

\begin{minipage}[b]{0.45\textwidth}
\begin{center}
\begin{align*}
{\rm HK}_{00}   \ce{->[k_1]} {\rm HK}_{p0} & \ce{->[k_2]} {\rm HK}_{0p} \ce{->[k_3]} {\rm HK}_{pp} 
\\
{\rm HK}_{0p} +{\rm RR} & \ce{->[k_4]} {\rm HK}_{00} +{\rm RR}_p \\
{\rm HK}_{pp} +{\rm RR} &  \ce{->[k_5] }{\rm HK}_{p0} +{\rm RR}_p\\ 
{\rm RR}_p & \ce{->[k_6]} {\rm RR}
\end{align*}

(a)
\end{center}
\end{minipage}
\begin{minipage}[b]{0.45\textwidth}
\begin{center}
\begin{tabular}
{c||c|c|c|c|c|c|c|c|c|c|}
 &   $x_1^0$ & $x_5^0$ \\\hline\hline
$\se'_1(0)$ &   \cellcolor{gray!30!white} $\pm$   & $+$  \\ \hline
$\se'_2(0)$ &   $+$ &   \cellcolor{gray!30!white} $\pm^*$   \\ \hline
$\se'_3(0)$ & $+$  & \cellcolor{gray!30!white} $\pm$   \\ \hline
$\se'_4(0)$ &  $+$ &    $-$  \\ \hline
$\se'_5(0)$ &   $-$ &   $+$ \\ \hline
$\se'_6(0)$ &   $+$ &   \cellcolor{gray!30!white} $\pm^*$ \\ \hline
\end{tabular}

\medskip
(b)
\end{center}
\end{minipage}
\caption{(a) A simple network of a hybrid histidine kinase, taken from \cite{feliu:unlimited}. (b) Sign-sensitivities with respect to an increase of $x_1^0$ and $x_5^0$. $\pm^*$ means that the sign is $+$ when $k_1\geq k_3$, that is, when the network has exactly one positive steady state. }\label{fig:HK}
\end{figure}

We order the species as HK, HK$_{p0}$,  HK$_{0p}$, HK$_{pp}$, RR and RR$_{p}$, and let $x_1,\dots,x_6$ denote  their concentrations respectively. Following \cite{FeliuPlos}, the set of positive steady states admits a parametrization in terms of $x_1$ and $x_5$ obtained by solving the steady states equations of $x_2,x_3,x_4,x_6$ in these variables:
\[x_{2}=\frac {k_{1}x_{1}( k_{4}x_{5}+k_{3}) }{k_{2}k_{4}x_{5}}, \quad 
x_{3}=\frac {k_{1}x_{1}}{k_{4}x_{5}},  \quad x_{4}=\frac{k_1k_3x_{1}}{k_{4}k_{5}x_{5}^{2}}, \quad 
x_{6}=\frac {k_{1}x_{1} ( k_{4}x_{5}+k_{3}) }{k_{4}k_6x_{5}}.\]
We choose the matrix of conservation laws 
$$ W=\begin{pmatrix} 1& 1 & 1 & 1 & 0 & 0 \\ 0 & 0 & 1 & 1 & 1 & 1 \end{pmatrix},$$
and construct $F_{k,x^0}(x)$ with first four components equal to $f_2,f_3,f_4,f_6$.
The determinant of $J_{\alpha,1}(x)$ evaluated at the parametrization yields:
\begin{align*}
\det J_{\alpha,1}(x_1,x_5) &= -(k_{1}-k_{3})k_{1}k_{2}k_{5}x_{1} 
-(k_{1}+k_2)k_{4}k_{5}k_{6}x_{5}^{2}
-k_{1}(k_{2}+k_3)k_{5}k_{6}x_{5}-k_{1}k_{2}k_{3}k_{6} \\ & \qquad -\frac{2\,k_{1}^{2}k_{2}k_{3}x_{1}}{x_{5}}-\frac {k_{1}^{2}k_{2}k_{3}^{2}x_{1}}{x_{5}^{2}k_{4}}.
\end{align*}
Here we see that if $k_1\geq k_3$, then this determinant has negative sign, which is actually the sign it attains when the steady state is asymptotically stable and hyperbolic. Indeed, in this case $n-d=4$ and $(-1)^\tau=(-1)^{n-1-1 + n-2-5} = (-1)^5=-1$.
If $k_3>k_1$, then the stable steady states will necessarily satisfy that the sign of $\det J_{\alpha,1}(x_1,x_5)$ is negative.
Using this, we proceed as above to compute the sign of the sensitivities with respect to adding a small amount to each of $x_i^0$. By Lemma~\ref{corollary}, it is enough to compute the sensitivities with respect to perturbing $x_1^0$ (which agrees with the perturbations with respect to $x_2^0$, $x_3^0$ and $x_4^0$) and $x_5^0$ (which agrees with $x_6^0$).

Figure~\ref{fig:HK}(b) shows the obtained sign-sensitivities under the assumption that $\det J_{\alpha,1}(x_1,x_5) $ is  negative. If this determinant is positive, then all signs are reversed, but this implies that the steady state is unstable. 

An apparently surprising property of this network is that the addition of $\HK_{00}$, that is, $x_1^0$, might lead to the decrease of $\HK_{00}$.
To have a closer inspection at this phenomenon, using the parametrization, we have that $\se_1'(0)$ at the steady state defined by $x_1, x_5$ is
$$\se_1'(0)= k_2k_5(k_1k_3x_1 - k_4k_6 x_5^2)  \, / \, \det J_{\alpha,1}(x_1,x_5). 
$$
By letting $k_1=\dots =k_6=1$, the system has exactly one positive steady state in each stoichiometric compatibility classes and $\det J_{\alpha,1}(x_1,x_5)<0$. 
The $x_1$-component of the steady state defined by $x_1=2,x_5=1$ will decrease after a small amount of $x_1$ is added to the system. On the other hand, the $x_1$-component of the steady state defined by $x_1=1,x_5=2$ will increase.

\medskip
\textbf{Two-site sequential and distributive phosphorylation cycle. } 
We conclude with one extra example where we analyze sign-sensitivities of a classical model. We consider 
the reaction network in which a substrate $S$ becomes doubly phosphorylated by a kinase $E$ and dephosphorylated by a phosphatase $F$. We let $S_0,S_1,S_2$ be the three phosphoforms of $S$ with $0,1,2$ phosphorylated sites, respectively.
This gives rise to the following reactions \cite{Wang:2008dc,conradi-mincheva}:
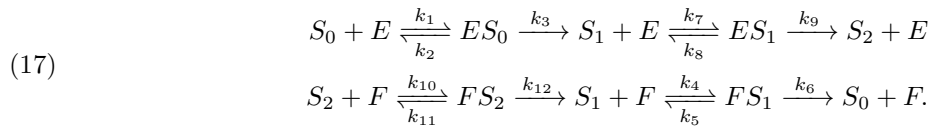
\begin{align}\label{eq:network}
\begin{split}
S_0 + E \ce{<=>[k_1][k_2]} ES_0 \ce{->[k_3]} S_1+E   \ce{<=>[k_7][k_8]} ES_1 \ce{->[k_9]} S_2+E \\
S_2 + F  \ce{<=>[k_{10}][k_{11}]} FS_2 \ce{->[k_{12}]} S_1+F  \ce{<=>[k_4][k_5]} FS_1 \ce{->[k_6]} S_0+F.
\end{split}
\end{align}
We order the species as $E,F,S_0,S_1,S_2,ES_0,FS_1,ES_1,FS_2$ and let $x_1,\dots,x_9$ denote their concentrations respectively.
A matrix of conservation laws is
$$ W = \begin{pmatrix}
1 & 0 & 0 & 0 & 0 & 1 & 0 & 1 & 0 \\
0 & 1 & 0 & 0 & 0 & 0 & 1 & 0 & 1 \\
0 & 0 & 1 & 1 & 1 & 1 & 1 & 1 & 1
\end{pmatrix}$$
The set of positive steady states admits a positive parametrization in $x_1,x_2,x_3$, obtained by solving the system $f_4=\dots=f_9=0$ in $x_4,\dots,x_9$, where $f_i$ is the mass-action evolution equation for $x_i$.  It is well known that this network admits between one and three positive steady states in each stoichiometric compatibility class, e.g. \cite{Wang:2008dc,conradi-mincheva}. 

We consider $\det J_{\alpha,1}(x)$ evaluated in the parametrization, and assume it is negative: namely this is the sign of this determinant when the steady state is asymptotically stable and hyperbolic. 
Under this assumption, we compute the sign-sensitivities $\se_1'(0),\dots,\se_5'(0)$ with respect to a small increase of $x_1^0$, $x_2^0$ and $x_3^0$, and obtain that none of them is given by a rational function with numerator of fix sign, indicating that all signs might be possible for this system. However, it is not straightforward to analyze the sign of the numerators while imposing that  $\det J_{\alpha,1}(x)$ is negative. 

Nevertheless, a few cases have a nice and simple form. Specifically:
\begin{itemize}
\item 
With respect to adding $x_3^0$, that is, adding substrate $S_0$, we have that  $\se_1'(0)>0$, $\se_2'(0)>0$ and $\se_4'(0)<0$ if  $k_3k_{12}\geq k_6k_9$. 
\item With respect to adding $x_2^0$, that is, adding phosphatase $F$, we obtain that 
 $\se_3'(0)<0$, if $k_3k_{12}\leq k_6k_9$, and $\se_5'(0)>0$ if   $k_3k_{12}\geq k_6k_9$.
\item 
Symmetrically, with respect to $x_1^0$, that is, adding kinase $E$, then $\se_3'(0)>0$, if $k_3k_{12}\leq k_6k_9$, and $\se_5'(0)<0$ if   $k_3k_{12}\geq k_6k_9$.
\end{itemize}

\section*{Acknowledgements }
This work has been partially supported by the Independent Research Fund of Denmark.  Janne Kool is thanked for her input and development of the ideas presented in this manuscript, and specially for pointing out the problem with perturbations with respect to total amounts.  \blue{Beatriz Pascual Escudero is thanked for comments on a preliminary version of this manuscript, and in particular for suggesting to include Proposition~\ref{prop:form}. }
%

\providecommand{\href}[2]{#2}
\providecommand{\arxiv}[1]{\href{http://arxiv.org/abs/#1}{arXiv:#1}}
\providecommand{\url}[1]{\texttt{#1}}
\providecommand{\urlprefix}{URL }

\end{document}